\newcommand{\extended}[1]{}    %
\newcommand{\short}[1]{#1}     %

\documentclass[sigconf]{aamas} %

\usepackage{balance} %

\usepackage{soul}
\usepackage{url}
\usepackage{adjustbox}
\usepackage{fancyvrb}

\usepackage{dcolumn}
\newcolumntype{d}[1]{D{.}{.}{#1} }

\usepackage{
	packages/wojtek15logics,
	packages/wojtek15other,
	packages/macros,
	packages/prooftree,
}
\SetAlFnt{\small\sffamily}
\renewcommand{\footnotetextcopyrightpermission}{}
\setcopyright{none}
\renewcommand{\acmConference}{}
\copyrightyear{2023}
\acmYear{2023}
\acmDOI{}
\acmPrice{}
\acmISBN{}

\title[Practical Abstraction for MAS]{Practical Abstraction for Model Checking of Multi-Agent Systems}

\author{Wojciech Jamroga}
\affiliation{
  \institution{University of Luxembourg}
  \city{Luxembourg}
  \country{Luxembourg}}
\email{wojciech.jamroga@uni.lu}

\author{Yan Kim}
\affiliation{
  \institution{University of Luxembourg}
  \city{Luxembourg}
  \country{Luxembourg}}
\email{yan.kim@uni.lu}

\begin{abstract}
Model checking of multi-agent systems (MAS) is known to be hard, both theoretically and in practice.
\extended{The state-space explosion is a major challenge here, as faithful models of real-world systems are immensely huge and infeasible even to generate -- let alone verify them.}%
A smart abstraction of the state space may significantly reduce the model, and facilitate the verification. %
\extended{%
However, while state abstraction is well studied from the theoretical point of view, little work has been done on how to define actual abstractions in practice.
}%
In this paper, we propose and study an intuitive agent-based abstraction scheme, based on the removal of variables in the representation of a MAS.
This allows to achieve a desired reduction of a state space without generating the global model of the system.
Moreover, the process is easy to understand and control even for domain experts with little knowledge of computer science.
We formally prove the correctness of the approach, and evaluate the gains experimentally on a family of a postal voting procedure models.
\end{abstract}


\definecolor{tucgreen}{RGB}{0,140,79}

\newcommand{\overapprox}{\precsim}

\newcommand{\simul}{\preceq}
\newcommand{\abstr}{\mathcal{A}}
\newcommand{\unwrap}{\mathcal{M}}
\newcommand{\Scope}{\textit{Sc}}
\newcommand{\dott}{\ .\ }

\newcommand{\masASV}{{\mathit{ASV}}}
\newcommand{\NCand}{\ensuremath{\mathit{NC}}}
\newcommand{\NVot}{\ensuremath{\mathit{NV}}}

\newcommand{\may}{\textit{may}}
\newcommand{\must}{\textit{must}}

\begin{document}

\pagestyle{fancy}
\fancyhead{}

\maketitle

\section{Introduction}\label{sec:intro}
\emph{Multi-agent systems (MAS)} describe interactions of multiple entities cal\-led {agents}, often assumed to be autonomous, intelligent, and/or rational. The theoretical foundations of MAS are mostly based on modal logic and game theory~\cite{Wooldridge02intromas,Shoham09MAS}. In particular, the temporal logics \CTL, \LTL, and \CTLs provide intuitive formalizations of many relevant properties, including reachability, liveness, safety, and fairness~\cite{Emerson90temporal}.
Algorithms and tools for verification\extended{ of such properties} have been in constant development for 40 years, with temporal model checking being the most popular approach~\cite{Baier08mcheck,Clarke18principles}.

However, formal verification of MAS is known to be hard, both theoretically and in practice.
The state-space explosion is a major source of the computational hardness, as faithful models of real-world systems are immensely huge and infeasible even to generate, let alone verify. In consequence, model checking of MAS with respect to their \emph{modular representations} ranges from \PSPACE-complete to undecidable~\cite{Schnoebelen03complexity,Bulling10verification}.
No less importantly, it is often unclear how to create the input model, especially if the system to be modeled involves human behavior~\cite{Jamroga20Pret-Uppaal}. Similarly, formalizing the relevant properties in the right way is by no means trivial~\cite{Jamroga21anticovid}. Both parts of the specification are error-prone and difficult to debug and validate: most model-checkers for MAS do not even have a graphical user interface.\footnote{
  Notable exceptions include \Uppaal~\cite{Behrmann04uppaal-tutorial} and STV~\cite{Kurpiewski21stv-demo}. }
Thus, in realistic cases, one does not even know if what they are verifying and what they \emph{think} they verify are indeed the same thing.

Much work has been done on containing the state-space explosion by smart representation and/or reduction of input models. Symbolic model checking based on SAT- or BDD-based representations of the state/transition space~\cite{McMillan93symbolic-mcheck,McMillan02unbounded,Penczek03ctlk,Kacprzak04verifying,Lomuscio07tempoepist,Huang14symbolic-epist} fall into the former group.
Model reduction methods include partial-order reduction~\cite{Peled93representatives,Gerth99por,Jamroga20POR-JAIR}, equivalence-based reductions~\cite{Bakker84equivalences,Alur98refinement,Belardinelli21bisimulations}, and state/action-abstraction~\cite{Cousot77abstraction,Clarke94abstraction,Godefroid02abstraction}. More specific variants of abstraction for multi-agent systems have been proposed in~\cite{Alfaro04three,Enea08abstractions,Cohen09abstraction-MAS,Lomuscio10dataAbstraction}, and more recently in~\cite{Kouvaros17predicateAbstraction,Belardinelli17abstraction,Belardinelli19abstractionStrat}.
A smart abstraction of the state space may reduce the model to manageable size by clustering ``similar'' concrete states into \emph{abstract states}, which should facilitate verification.
Unfortunately, such clustering may remove essential information from the model, thus making the verification of the abstract model inconclusive for the original model.
Lossless abstractions can be obtained by means of abstraction-refinement~\cite{Clarke00cegar,Clarke03cegar,Shoham04abstraction,Ball06abstraction,Dams18abstraction+refinement} but, typically, they are difficult to compute or provide insufficient reduction of the model -- quite often both.

In consequence, one has to live with abstractions that only approximate the concrete model. Moreover, crafting a good abstraction is an art that relies on the domain expertise of the modeler.
Since domain experts are seldom computer scientists or specialists in formal methods, the theoretical formulation of abstraction as an arbitrary mapping from the concrete to the abstract state space has little appeal.
Moreover, model checking tools typically do not support abstraction, so doing one would require to manipulate the input specification code, which is a difficult task in itself.
What we need is a simple and intuitive methodology for selecting information to be removed from a MAS model, and for its automated removal that preserves certain guarantees.
Last but not least, practical abstraction should be applied on modular representations of MAS, unlike the theoretical concept that is usually defined with respect to explicit models of global states.

In this paper, we suggest that the conceptually simplest kind of abstraction consists in removing a domain variable from the input model.
This can be generalized to the merging of several variables into a single one, and possibly clustering their valuations. 
It is also natural to restrict the scope of abstraction to a part of the graph.
As the main technical contribution, we propose a correct-by-design method to generate such abstractions.
We prove that the abstractions preserve the valuations of temporal formulae in Universal \CTLs (\ACTLs).
More precisely, our \may-abstractions preserve the falsity of \ACTLs properties, so if $\varphi\in\ACTLs$ holds in the abstract model, it must also hold in the original one.
Conversely, our \must-abstractions preserve the truth of \ACTLs formulae, so if $\varphi\in\ACTLs$ is false in the abstract model, it must also be false in the original one.
We also evaluate the efficiency of the method by verifying a scalable model of postal voting in \Uppaal. 
The experiments show that the method is user-friendly, compatible with a state of the art verification tool, and capable of providing significant computational gains.

\section{Preliminaries}\label{sec:preliminaries}

We start by introducing the models and formulae which serve as an input to model checking.

\subsection{MAS Graphs}

To represent the behavior of a multi-agent system, we use modular representations inspired by reactive modules~\cite{Alur99reactive}, interleaved interpreted systems~\cite{lomuscio10partialOrder,Jamroga20POR-JAIR}, and in particular by the way distributed systems are modeled in \Uppaal~\cite{Behrmann04uppaal-tutorial}.

Let $\Var$ be a finite set of typed variables over finite domain%
\footnote{We will focus exclusively on variables over finite domains, the only variant of variables supported by Uppaal, and, naturally, the one interpreted by the programs.}.
By $\Eval(\!\Var)$ we denote a set of evaluation functions which map variables $v\in\Var$ to values in their domains $\Dom(v)$.

$\Cond$ is a set of logical expressions (also called \emph{guards}) over $\Var$ and literal values, possibly involving arithmetic operators (i.e., ``$+$'', ``$-$'', ``$*$'', ``$/$'' and ``$\%$'').

Let $\ChanId$ be a finite set of asymmetric synchronization channels. We define the set of synchronizations as $\Chan = \set
{c!, c? \mid c\in\ChanId}\cup\set{-}$,
with $c!$ and $c?$ for sending and receiving on a channel $c$ accordingly, and ``$-$'' standing for no synchronization.
For the paper, we will focus only on 1-to-1 type of synchronization; adding the 1-to-many (broadcast) and many-to-many is straightforward (by minor tweaks of \autoref{def:combined}), is of a less academic value, and therefore omitted.

\begin{definition}
	An \emph{\agname} \extended{, describing the behaviour of an agent,} is a tuple
	$\agsym=(\Var, \Loc, l_0, g_0, \textit{Act},$ $\textit{Effect}, \hookrightarrow)$,
	consisting of:
	\begin{itemize}[noitemsep,parsep=0pt]
	\item $\Var$: a finite set of typed \emph{variables} over finite domains;
	\item $\Loc$: a non-empty finite set of \emph{locations} (nodes in the graph);
	\item $l_0\in\Loc$: the initial location;
	\item $g_0\in \Cond$: the initial condition\extended{, required to be satisfied at the initial location};
	\item $\Act$: a set of \emph{actions}, with $\tau\in\Act$ standing for ``do nothing'';
	\item $\textit{Effect}: \Act\times \Eval(\!\Var) \mapsto \Eval(\!\Var)$: the \emph{effect} of an action\footnote{%
		Clearly, on an argument  $\tau$ it will behave just as an identity function over $\Eval(\!\Var)$.%
	};
	\item $\hookrightarrow\subseteq\Loc\times\Label\times\Loc$: the set of labelled edges defining the local \emph{transition relation}.
	\end{itemize}
\end{definition}

The set of labels is defined as $\Label\subseteq\Cond \times\Chan \times\Act$.
\extended{Therefore, we introduce destructors $\textit{cond}$, $\textit{sync}$, $\textit{act}$ mapping a label to its counterpart.}
An edge $(l,T,l')\in\,\hookrightarrow$, where $T\in\Label$, will often be denoted by $l\xhookrightarrow{g:ch\,\alpha}l'$, where $g=\textit{cond}(T)$, $ch=\textit{sync}(T)$ and $\alpha=\textit{act}(T)$.\footnote{%
	For the sake of readability, when $ch=-$, the channel sub-label will be omitted.
}%

A condition $g\in\Cond$ can be associated with a set of evaluations $\textit{Sat}(g)=\{\eta\in\Eval(\!\Var) \ |\ \eta\models g \}$ satisfying it.
An edge labelled with $T\in\Label$ is said to be \emph{locally enabled} for evaluation $\eta\in\Eval(\!\Var)$ iff
$\eta\in\textit{Sat}(\cond(T))$.
Furthermore, every action $\alpha\in\Act\setminus\set{\tau}$ can be associated with a non-empty sequence of atomic assignment statements (also called \emph{updates}) of the form $\alpha^{(1)}\alpha^{(2)}\ldots\alpha^{(m)}$.%

A function $\mathcal{V}:\!\Cond\cup\!\Act\mapsto 2^{\Var}$ maps an expression or an assignment statement to the set of occurring variables.

Without loss of generality, we assume that for the set $\Var=\set{v_1,\ldots,v_k}$ its elements are enumerated and/or ordered in some arbitrary yet fixed way.
Then the evaluation of any $V\subseteq\Var$ can be represented as a vector $\eta(V)=[\eta(v_{i_1}),\ldots,\eta(v_{i_l})]$, where $\forall_{j=1..l}.i_j\in\set{1,\ldots,k}\wedge \forall_{j=1..l}.i_j\leq i_{j+1}$.

Let $V\subseteq\!\Var$ and $r\in\Cond\cup\Act\cup\Eval(\!\Var)$, by $r[V=C]$ we denote the substitution of all free occurrences of variables $V$ in $r$ by $C\in\Dom(V)$.

\begin{definition}
A \emph{\magname}\ is a multiset of \agnames additionally parameterized by a set of shared (global) variables.
We assume w.l.o.g.~that all local variables have unique names (e.g., achieved by prefixing those with a name of its agent graph), then a set of shared variables can be derived by taking an intersection over the variable sets of component agents.
\end{definition}
	
\begin{figure}[t]
	\centering
	\begin{subfigure}[c]{0.45\columnwidth}
		\centering
		\includegraphics[width=\linewidth]{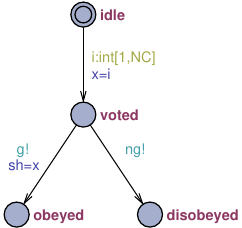}
		\label{fig:upp-async-voting-v}
	\end{subfigure}
	\begin{subfigure}[c]{0.53\columnwidth}
		\centering
		\includegraphics[width=\linewidth]{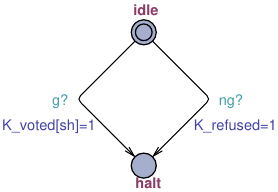}
		\label{fig:upp-async-voting-c}
	\end{subfigure}
	\caption{\magname\ for ASV: (a) Voter, (b) Coercer}
	\label{fig:upp-async-voting}
\end{figure}

\begin{example}[ASV]
\label{ex:working-model}
As the running example, we use a variation of the Asynchronous Simple Voting scenario of~\cite{Jamroga20POR-JAIR}.
Its \magname\ $\masASV=\multiset{\Var_{sh},\agsym_{\textit{Voter}},\agsym_{\textit{Coercer}}}$ is shown in~\autoref{fig:upp-async-voting}.
The system is parameterized by the number of candidates $NC$ and consists of a voter graph $\agsym_{\textit{Voter}}$, a coercer graph $\agsym_{\textit{Coercer}}$, a set of shared variables $\Var_{sh}=\{sh\}$ and an implicit $g_0 = \bigwedge_{v\in\Var}(v==0)$.

The voter starts by non-deterministically selecting one of the candidates (\texttt{i:int$[$1,NC$]$}), and casting a vote for the candidate ({\small{\texttt{idle}$\to$\texttt{voted}}}).
Then, she decides to either give the proof of how she voted to the coercer
({\small\texttt{voted}$\to$\texttt{obeyed}}), %
or to refuse it ({\small\texttt{voted}$\to$\texttt{disobeyed}}). %
Both options require executing a synchronous transition with the coercer.
In turn, the coercer either gets the proof and learns for whom the vote was cast, or becomes aware of the voter's refusal.
\end{example}

\subsection{Models of MAS Graphs}

\begin{figure*}%
	\centering
	\begin{tabular}{@{}c@{\hspace{0.9cm}}c@{}}
    \begin{tabular}{@{}c@{}}
    \includegraphics[width=0.7\columnwidth]{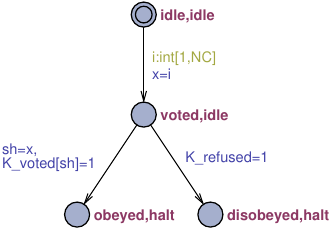}
    \end{tabular}
      &
    \begin{tabular}{@{}c@{}}
    \includegraphics[width=1.3\columnwidth]{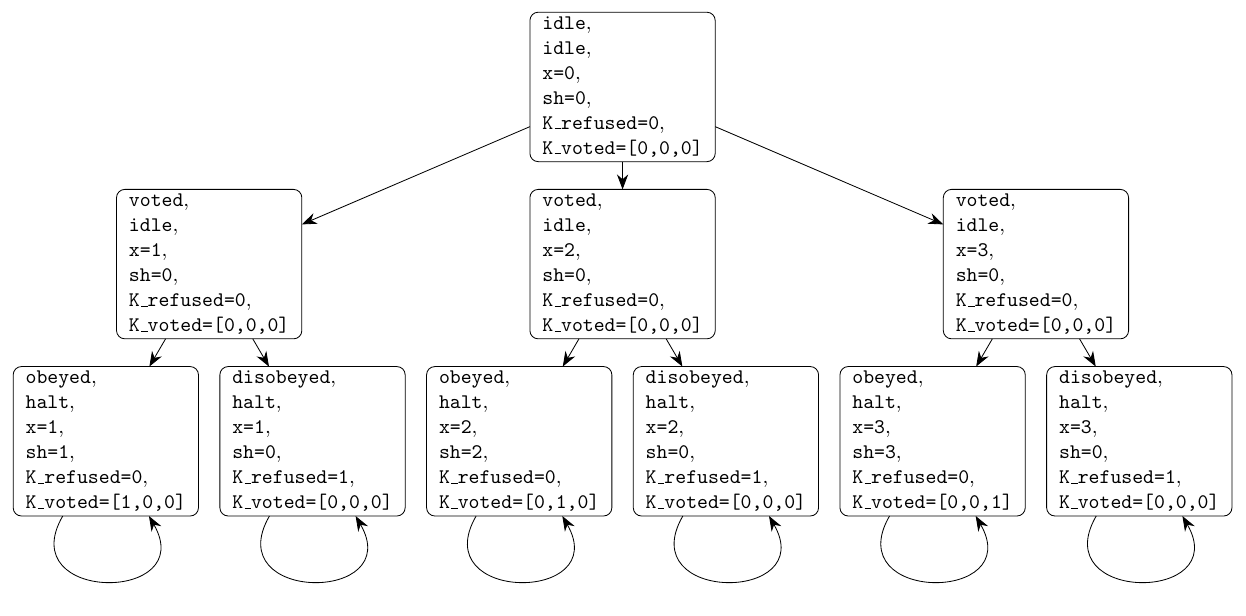}
    \end{tabular}
    \end{tabular}
    \vspace{-0.4cm}
	\caption{Models of MAS:\quad (a) combined MAS graph of ASV;\quad (b) Unwrapping for ASV with \mbox{$\NCand=3$} candidates}
	\label{fig:upp-async-voting-composition}
	\label{fig:upp-async-voting-lts}
\end{figure*}

We define the execution \extended{semantics }of a MAS graph by its \emph{unwrapping}.

\begin{definition}
	\label{def:combined}
	Let $\magsym = \multiset{\Var_{sh},\agsym_1,\ldots,\agsym_n}$ be a
	\magname having a set of shared variables $\Var_{sh}$.
	The \emph{combined \magname} of $\magsym$ is the agent graph $\agsym_{\magsym} = (\Var, \Loc, l_0, g_0, \Act, \Effect,\hookrightarrow)$, where $\Var = \bigcup_{i=1}^{n}\Var_i$, $\Loc = \Loc_1\times\ldots\times\Loc_n$, $l_0 = (l_{0,0}, \ldots, l_{n,0})$, 
	$g_0 = g_{0,0}\wedge\ldots\wedge g_{n,0}$, $\Act = \bigcup_{i=1}^{n}\Act_i$.

	Relation $\hookrightarrow$ is obtained inductively by the following rules (where $l_i,l'_i\in\Loc_i$, $l_j,l_j'\in\Loc_j$, $c\in \ChanId_i\cap\ChanId_j$
	for two	\agnames $\agsym_i$ and $\agsym_j$ of distinct indices	$1 \leq i,j\leq n$):

	\begin{small}
	\begin{tabular}{c@{\qquad}c}
	$\begin{prooftree}
			l_i \lhook\joinrel\xrightarrow{g_i:c!\alpha_i}_i l_i' \wedge l_j \lhook\joinrel\xrightarrow{g_j:c?\alpha_j}_j l_j'
			\justifies
			(l_i,l_j)\lhook\joinrel\xrightarrow{g_i\wedge g_j: (\alpha_j\updcomp \alpha_i) }(l_i',l_j')
			\end{prooftree}$	
	&
	$\begin{prooftree}
		l_i \lhook\joinrel\xrightarrow{g_i:\alpha_i}_i l_i'
		\justifies
		(l_i,l_j)\lhook\joinrel\xrightarrow{g_i: \alpha_i }(l_i',l_j)
		\end{prooftree}$
	\\
	\\
	$\begin{prooftree}
			l_i \lhook\joinrel\xrightarrow{g_i:c?\alpha_i}_i l_i' \wedge l_j \lhook\joinrel\xrightarrow{g_j:c!\alpha_j}_j l_j'
			\justifies
			(l_i,l_j)\lhook\joinrel\xrightarrow{g_i\wedge g_j: (\alpha_i\updcomp \alpha_j) }(l_i',l_j')
			\end{prooftree}$
	&
	$\begin{prooftree}
		l_j \lhook\joinrel\xrightarrow{g_j:\alpha_j}_j l_j'
		\justifies
		(l_i,l_j)\lhook\joinrel\xrightarrow{g_j: \alpha_j }(l_i,l_j')
		\end{prooftree}$
	\end{tabular}
	\end{small}

	\noindent Lastly, the effect function is defined by: %
	$$\small\textit{Effect}(\alpha,\eta) =
	\begin{cases}
		\Effect_i(\alpha,\eta) &\text{if } \alpha\in \Act_i\\
		\Effect(\alpha_i,\Effect(\alpha_j,\eta)) &\text{if } \alpha = \alpha_i\updcomp\alpha_j
	\end{cases}$$

	Note that by construction of \extmagname its edges can be labelled only by $T\in\Label$ where $\textit{sync}(T)=-$.
\end{definition}

\begin{example}\label{ex:async-voting-concurrent-composition}
The \extmagname	$\agsym_{\masASV}$ for asynchronous simple voting of \autoref{ex:working-model} is depicted in \autoref{fig:upp-async-voting-composition}a.

\end{example}

\begin{definition}
A \emph{model} is a tuple $M=(\textit{St}, I, \longrightarrow, AP, L)$, where
	$\textit{St}$ is a set of states,
	$I\subseteq \textit{St}$ is a non-empty set of initial states,
	$\longrightarrow  \subseteq \textit{St}\times \textit{St}$ is a transition relation,
	$AP$ is a set of atomic propositions,
	$L:\textit{St}\to 2^{AP}$ is a labelling function.
    We assume $\longrightarrow$ to be serial, i.e., there is at least one outgoing transition at every state.
\end{definition}

Nodes and edges in an agent graph correspond to \emph{sets} of states and transitions, defined by its unwrapping.

\begin{definition}\label{def:unwrapping}
The \emph{\unwrapping} of an \agname\ $\agsym%
\extended{=(\textit{Var}, \textit{Loc}, l_0, \textit{Cond}, g_0, \hookrightarrow, \textit{Act}, \textit{Effect})}$ %
is a model $\mathcal{M}(\agsym) = (\textit{St}, I, \longrightarrow, AP, L)$, where:
\begin{itemize}
	\item $\textit{St} = \Loc \times \Eval(\!\Var)$,
	\item $I = \{ \abracket{l_0, \eta} \in \textit{St}\mid \eta\in \textit{Sat}(g_0) \}$,
	\item $\longrightarrow\ = \longrightarrow_0\ \cup \{(s,s)\in \textit{St} \times \textit{St} \mid \lnot\exists s'\in \textit{St}\dott s\longrightarrow_0 s'\}$, where
        $\longrightarrow_0\ =\{(\abracket{l,\eta},\abracket{l',\eta'})\in \textit{St} \times \textit{St} \mid \exists\ l\xhookrightarrow{g:\alpha}l'\,.\,\eta\in \textit{Sat}(g) \wedge \eta'=\Effect(\alpha, \eta)\}$,%
		\footnote{Note that $\longrightarrow$ adds loops wherever necessary to make the relation serial.}%
	\item $AP = \Loc \cup \!\Cond$,
	\item %
	$L(\abracket{l,\eta})=\{l\}\cup\{g\in \Cond \mid \eta \in \textit{Sat}(g)\}$.
\end{itemize}

The unwrapping $\unwrap(\magsym)$ of a \magname $\magsym$ is given by the unwrapping of its combined graph.
\end{definition}

By a slight abuse of notation, we lift some operators to sets of states.
For example, $L(S) = \bigcup_{s\in S} L(s)$ for $S\subseteq\textit{St}$.
Moreover, we will use $AP(V)$ to denote the subset of propositions that do not use variables from outside $V$.

\begin{example}\label{ex:asv-unwrapping}
The unwrapping $\mathcal{M}({\masASV})$ of the \magname for asynchronous simple voting (\autoref{ex:working-model})
with 3 candidates is shown in \autoref{fig:upp-async-voting-lts}b.
\end{example}

\begin{definition}\label{def:path}
Let $M$ be a model.
A \emph{run} in $M$ is a sequence of states $s_0s_1\dots$, such that $s_i\in\textit{St}$ and $s_i\longrightarrow s_{i+1}$ for every $i$.
{A \emph{path} is an infinite run.}
For a finite run $\pi=s_0s_1\ldots s_n$, let $\len(\pi)=n$ denote its length.
By $\pi[k]$ and $\pi[i,j]$ we will denote a $k$-th state and a run fragment from $i$ to $j$ of $\pi$ accordingly. %
The sets of all runs in $M$, all paths in $M$, and all paths starting from state $s$ are denoted by
$\textit{Runs}(M)$, $\textit{Paths}(M)$, and $\textit{Paths}(s)$.
Similarly, $\textit{Runs}^{t}$ will denote the set of runs of fixed length $t\in\mathbb{N}^+\cup\{\infty\}$.

A state $s\in\textit{St}$ is said to be \emph{reachable} in model $M$ if there exists a run $\pi=s_0\ldots s_n$, s.t. $s_0\in I$ and $s_n=s$.
\end{definition}

\extended{
  \begin{definition}
  	Let $M=\mathcal{M}(\magsym)$ be a model of an \extmagname\ $\magsym$.
  	A
  	sequence of states
  	$\pi = \abracket{l_`1,\eta_1}\rightarrow\abracket{l_2,\eta_2}\rightarrow\ldots$
  	is called a \emph{run} of model $M$ if following holds:
  	\begin{itemize}
  		\item if $l_1=l_0$, then $\eta_1\in\textit{Sat}(g_0)$,
  		\item $\Apath  {1\leq i< \len(\pi)}\;\exists\ l_i\xhookrightarrow{g_i:\alpha_i}l_{i+1}$ such that:
  		\begin{itemize}
  			\item[--] $\eta_i\in\textit{Sat}(g_i)$, and
  			\item[--] $\eta_{i+1}=\textit{Effect}(\alpha_{i}, \eta_i)$ 
  		\end{itemize}
  	\end{itemize}
  	where $\len(\pi)\in\mathbb{N}^{+}\cup\{ \infty \}$.
  	\smallskip

  \end{definition}

  \begin{definition}
  	A \emph{path} is a maximal run, in a sense that it is either infinite or ends in a state with no outgoing transitions.
  	$\textit{Paths}(M)\subseteq\textit{Runs}(M)$ is a set of all paths of $M$,
  	whereas
  	$\textit{Paths}(s)\subseteq\textit{Paths}(M)$
  	is a set of paths starting in state $s\in\textit{St}$.
  \end{definition}

\begin{definition}
	An evaluation $\eta\in \Eval(\!\Var)$ is \emph{reachable} at location $l\in\Loc$ if there exists a run $\pi\in \textit{Runs}(M(P))$ such that $\pi[j] = \abracket{l,\eta}$ for some $0\leq j\leq \len(\pi)$.
\end{definition}
}%

\subsection{Branching-Time Logic \actls}

To specify requirements, we use the \emph{universal fragment} of the branching-time logic \CTLs~\cite{Emerson90temporal} (denoted $\actls$)%
\footnote{%
	Not to be confused with ``Action \CTL'' of~\cite{Nicola90action-ctl}.%
}
with $\Apath$ (``for every path'') as the only path quantifier.
The syntax for $\actls$ over a set of atomic propositions $AP$ is formally given by:
\begin{eqnarray*}
\psi &::=& \top \mid  \bot\mid a\mid \neg a \mid \psi \wedge \psi\mid \psi \vee \psi \mid \Apath \varphi \\
\varphi &::=& \psi\mid \varphi\wedge\varphi \mid \varphi\vee\varphi \mid \Next\varphi \mid \varphi\Until\varphi 
\end{eqnarray*}
where $a\in AP$, and $\Next,\!\Until$ stand for ``next'' and ``until,'' respectively.
Formulae $\psi$ are called state formulae, and $\varphi$ are called path formulae.
The semantics of $\actls$ is given with respect to states $s$ and paths $\pi$ of a model $M$.
\extended{
	Let $M=(\textit{St}, I, \rightarrow, AP, L)$, $a\in AP$.
	Satisfaction relation $\models$ of a state formula $\psi$ for a state $s\in\textit{St}$ (denoted by $s\models\psi$) is defined by:
}%
\begin{alignat*}{3}
&M,s\models a &&\quad\text{iff }a\in L(s)\\
\extended{
	&M,s\models \neg a &&\quad\text{iff } a \notin L(s) \\
	&M,s\models \psi_1\wedge\psi_2 &&\quad\text{iff } M,s\models\psi_1 \wedge M,s\models\psi_2 \\
	&M,s\models \psi_1\vee\psi_2 &&\quad\text{iff } M,s\models\psi_1 \vee M,s\models\psi_2 \\
}%
&M,s\models \Apath\varphi &&\quad\text{iff } M,\pi\models\varphi\text{ for all }\pi\in\textit{Paths}(s) \\
&M,\pi\models \psi &&\quad\text{iff }M,\pi[0]\models\psi\\
&M,\pi\models \Next\varphi &&\quad\text{iff }M,\pi[1,\infty]\models\varphi\\
&\M,\pi\models \varphi_1\Until\varphi_2 &&\quad\text{iff } \exists j\,.\,(M,\pi[j,\infty]\models\varphi_2\; \wedge\\
& &&\quad\hspace*{\fill} \forall 1\leq i<j\,.\,M,\pi[i,\infty]\models\varphi_1)
\extended{\\
    &M,\pi\models \varphi_1\wedge\varphi_2 &&\quad\text{iff } M,\pi\models \varphi_1 \wedge M,\pi\models \varphi_2 \\
    &M,\pi\models \varphi_1\vee\varphi_2 &&\quad\text{iff } M,\pi\models \varphi_1 \vee M,\pi\models \varphi_2
}%
\end{alignat*}
\short{The clauses for Boolean connectives are standard. }%
Additional (dual) temporal operators ``sometime'' and ``always'' can be defined as $\Sometm\psi\equiv\top\Until\psi$ and $\Always\psi\equiv\neg\Sometm\neg\psi$.
Model $M$ satisfies formula $\psi$ (written $M\models \psi$) iff $M,s_0\models \psi$ for all $s_0\in I$.

\begin{example}
Model $M=\mathcal{M}(\masASV)$ in \autoref{fig:upp-async-voting-lts}b
satisfies the $\actls$ formulae
$\Apath \Always (\neg\texttt{obeyed}\vee \texttt{K\_voted[x]=1})$,
saying that if Voter obeys, a Coercer gets to know how she voted, and
$\Apath \Always (\neg\texttt{disobeyed}\vee \texttt{K\_refused=1})$,
saying that she cannot disobey Coercer's instructions without his knowledge.
It does not satisfy
$\Apath \Sometm (\neg\texttt{K\_voted=[0,0,0]})$,
saying that Coercer will eventually get to know how Voter voted. %
\end{example}

\section{\mbox{Variable Abstraction for MAS Graphs}}\label{sec:abstraction}

In this section, we describe how to automatically reduce MAS graphs by simplifying their local variables.

\begin{algorithm}[t]
    \SetAlgoLined\DontPrintSemicolon
    compute the combined graph $\agsym_{\magsym}$ of $\magsym$\;
    compute the approximate local domain $d$ for $V$ in $\agsym_{\magsym}$\;
    \ForEach{agent graph $\agsym_i\in \magsym$}{
        compute abstract graph $\abstr(\agsym_i)$ w.r.t. $\Var_i\cap V$ and $d_i$\;
    }
    \Return $\abstr(\magsym)=\multiset{\Var_{sh},\abstr(\agsym_1),\ldots,\abstr(\agsym_n)}$
    \caption{Abstraction of MAS graph $\magsym=\multiset{\Var_{sh},\agsym_1,\ldots,\agsym_n}$ w.r.t. variables $V$}
\label{alg:abstraction-idea}
\end{algorithm}

\subsection{Main Idea}

We transform the MAS graph $\magsym=\multiset{\Var_{sh},\agsym_1,\ldots,\agsym_n}$ in such a way that:\
(i) the resulting abstract MAS graph $\abstr(\!\magsym)$ is transformed \emph{locally}, i.e., $\abstr(\magsym)=\multiset{\Var_{sh},\abstr(\agsym_1),\ldots,\abstr(\agsym_n)}$;\
(ii) the abstract agent graphs $\abstr(\agsym_i)$ have the same structure of locations as their concrete versions $\agsym_i$;\
(iii) the only change results from removal of a subset of local variables $V$, or simplifying their domains of values.

Moreover, we want may-abstraction $\abstr^{\may}=\abstr(\!\magsym)$ to \emph{over-approximate} $\magsym$, in the sense that every transition in $\magsym$ has its counterpart in $\abstr^{\may}$.
Then, every formula of type $\Apath\varphi$ that holds in the model $\unwrap(\abstr^{\may})$ must also hold in the model $\unwrap(\magsym)$.

Alternatively, we may want must-abstraction $\abstr^{\must}=\abstr(\magsym)$ to \emph{under-approximate} $\magsym$, in the sense that all transitions in $\abstr^{\must}$ have their counterparts in $\magsym$.
Then, if $\Apath\varphi$ is false in $\unwrap(\abstr^{\must})$, it also cannot be true in $\unwrap(\magsym)$.

The general structure of the procedure is shown in~\autoref{alg:abstraction-idea}.
First, we approximate the set of reachable evaluations $d(V,l)$ in every location of a combined MAS graph $\agsym_{\magsym}$.
This can be done following the algorithm \approxLocalDomain, presented in Appendix~\ref{sec:approxDomain}.
Then, its output is used to transform the agent graphs $\agsym_i$, one by one, by removing (or simplifying) those variables of $V$ that occur in $\agsym_i$.
We implement it by function \generalVarAbstraction in Sections~\ref{sec:abstraction-by-removal}--\ref{sec:restricted-scope-abstraction}.

\para{Assumptions and Notation.}
Assume there is a unique initial evaluation $\eta_0$ satisfying initial condition $g_0$, $\exists!\eta_0\in\Eval(\!\Var)\cap\textit{Sat}(g_0)$, and assigning each\extended{ variable} $v\in\Var$ with its default value $v_0=\eta_0(v)$.

Furthermore, for an approximation $d$ obtained from variables $V$ and $\Loc=\Loc_1\times\ldots\times\Loc_n$, by $d_i$ we denote a reduced on the $i$-th location component ``narrowing'' of a local domain, where $1\leq i\leq n$.
Intuitively, for $l_j\in\Loc_i$ the value of $d_i(V, l_j)$ is defined as $\bigcup_{l\in \Loc_1\times\ldots\times\Loc_{i-1}\times\set{l_j}\times\Loc_{i+1}\times\ldots\times\Loc_n}d(V,l)$.

\para{Upper- vs. Lower-Approximation.}
We define two variants of the algorithm.
The upper-approximation of local domain for every $l\in\Loc$ initializes $d(V,l)=\varnothing$, and then adds new, possibly reachable values of $V$ whenever they are produced on an edge coming to $l$.
The lower-approximation  initializes $d(V,l)=\Dom(V)$, and iteratively removes the values might be unreachable.
To this end, \approxLocalDomain is parameterized by symbols $d_0$ and $\auxop$, such that $d_0=\varnothing$ and $\auxop=\cup$ for the upper-approximation, and $d_0=\Dom(V)$ and $\auxop=\cap$ for the lower-approximation.
Note that $d_0$ is simply a neutral (identity) element of the $\auxop$-operation.

{\small
\begin{table}[t]
	\centering
	\begin{tabular}{lrr}
	\toprule
	$l\in\Loc$  & $r(l)$ & $d(x,l)$ \\
	\midrule
	$\abracket{\texttt{idle,idle}}$      & $3$  & $\{0\}$     \\
	$\abracket{\texttt{voted,idle}}$     & $2$  & $\{1,2,3\}$      \\
	$\abracket{\texttt{obeyed,halt}}$    & $0$  & $\{1,2,3\}$      \\
	$\abracket{\texttt{disobeyed,halt}}$ & $0$  & $\{1,2,3\}$      \\
	\bottomrule
	\end{tabular}
	\caption{Reachability index $r$ and upper-approximation domain $d$ for ASV}
	\label{tab:example-aux-results}
\end{table}
}

\begin{example}\label{ex:domainapprox}
The output of \approxLocalDomain\ for the upper-approximation of variable $x$ in the combined ASV graph of \autoref{ex:async-voting-concurrent-composition} can be found in \autoref{tab:example-aux-results}.
\end{example}

\subsection{Abstraction by Removal of Variables}\label{sec:abstraction-by-removal}

\begin{algorithm}[t]
  \SetAlgoLined\DontPrintSemicolon
  \myproc{\computeAbstraction{$\agsym=\agsym_i$, $V$, $d=d_i$}}{
  $\hookrightarrow_a:=\varnothing$\;
  \ForEach{$l\xhookrightarrow{g:ch\,\alpha} l'$}{
  \ForEach{$c\in d(V,l) $}{
  $g':=g[V=c]$\; %
  $\delta_0=\set{\eta\in\mathit{Sat}(g)\mid \eta(V)=c}$\;
  \text{let }$\alpha'=\alpha^{(1)}\ldots\alpha^{(m)}$\;
  \For{$i=1$ \KwTo $m$}{
    $\delta_i=\set{\eta'=\Effect(\eta,\alpha^{(i)})\mid \eta\in\delta_{i-1}}$\;
    \uIf{ $\LHS(\alpha^{(i)})\notin V$}{
      $\alpha^{(i)}:=\tau$
    }
  }
  $A=\prod_{i=1}^{m}\set{\alpha^{(i)}[V=\eta(V)]\mid \eta\in\delta_i}$\;
  $\hookrightarrow_a := \hookrightarrow_a \cup (\bigcup_{\alpha'\in A} \{ l\xhookrightarrow{g':ch\,\alpha'} l' \} )$
  }
  }
  $\agsym.\agedges := \hookrightarrow_a$\;
  \KwRet{$\agsym$}
  }
  \caption{Abstraction by variable removal}
  \label{alg:over-approx-abstraction}
\end{algorithm}

The simplest form of variable-based abstraction consists in the complete removal of a given variable $x$ from the MAS graph.
To this end, we use the approximation of reachable values of $V\subseteq \Var$, produced by \approxLocalDomain.
More precisely, we transform every edge between $l$ and $l'$ that includes variables from $V'\subseteq V$ in its guard and/or its update into a set of edges (between the same locations), each obtained by substituting $V'$ with a different value $C\in d(V',l)$, see \autoref{alg:over-approx-abstraction}.
The abstract agent graph obtained by removing variables $V$ from $\agsym$ in the context of $\magsym$ is denoted by $\mathcal{A}_{\set{V}}\left(\agsym,\magsym\right)$.
Whenever relevant, we will use $\mathcal{A}^{\may}$ (resp.~$\mathcal{A}^{\must}$) to indicate the used variant of abstraction.

\begin{example}
The result of removing variable $x$ from the voter graph, according to the domain's upper-approximation presented in \autoref{tab:example-aux-results}, is shown in \autoref{fig:upp-async-voting-composition-abstract}. Note that its unwrapping (\autoref{fig:upp-async-voting-abstract-lts}) is distinctly smaller than the original one (\autoref{fig:upp-async-voting-lts}b). Still, as we will formally show in Section~\ref{sec:correctness}, all the paths of the model in \autoref{fig:upp-async-voting-lts}b are appropriately represented in the model of \autoref{fig:upp-async-voting-abstract-lts}.
\end{example}

\subsection{Merging Variables and Their Values}\label{sec:abstraction-by-merging}

\begin{figure}[t]
	\centering
	\includegraphics[width=0.45\columnwidth]{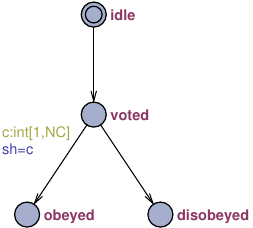}
	\caption{Over-approx. abstraction \mbox{$\abstr^{\may}_{\set{x}}(\agsym_{\textit{Voter}}, \masASV)$}}%
	\label{fig:upp-async-voting-composition-abstract}
\end{figure}

\begin{figure}[t]
	\centering
	\includegraphics[width=\linewidth]{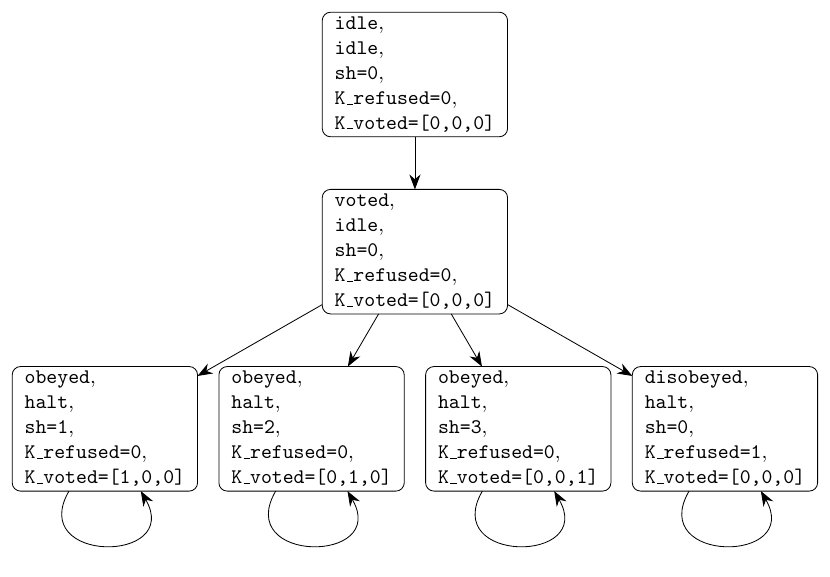}
    \vspace{-0.8cm}
	\caption{Unwrapping for the over-approximating abstraction $\abstr^{\may}_{\set{x}}(\masASV)\ =\ \multiset{\Var_{sh},\ \abstr^{\may}_{\set{x}}(\agsym_{\textit{Voter}},\masASV),\ \agsym_{\textit{Coercer}})}$}
	\label{fig:upp-async-voting-abstract-lts}
\end{figure}

A more general variant of variable abstraction assumes a collection of mappings $\mapf=\{f_1,\ldots,f_m\}$.
Each mapping $f_i:\Eval(X_i)\mapsto\Eval(z_i)$ merges the local variables $X_i\subseteq\Var_j$ of some agent graph $\agsym_j$ to a fresh variable $z_i$.
The abstraction based on $f_i$ removes variables $X_i$ from graph $\agsym_j$, and replaces them with $z_i$ that ``clusters'' the values of $X_i$ into appropriate abstraction classes.
We will use $\Argsd(f_i)=X_i$ and $\Argsd(\mapf)=\bigcup_{i=1}^{m}\Argsd(f_i)$ to refer to the variables removed by $f_i$ and $\mapf$.
$\Argsr(f_i)=\{z_i\}$ and $\Argsr(\mapf)=\bigcup_{i=1}^{m}\Argsr(f_i)$ refer to the new variables.

Note that the procedure in Section~\ref{sec:abstraction-by-removal} can be seen as a special case, with a sole mapping $f$ merging $X=\set{x}$ to a fresh variable $z$ with the singleton domain $\Dom(z)=\set{\eta_0(x)}$.

\extended{

\begin{algorithm}[t]
    \SetAlgoLined\DontPrintSemicolon

    \myproc{\computeMergeAbstraction{$\agsym$, ${d}\triangleright \agsym.\Var$, $\mapf$}}{
      $Y:=\Argsd(\mapf)\cap \agsym.\Var$\;
      $Z:=\Argsr(\mapf)$\;
      $g_0:=g_0\wedge(Z=f(\restr{\eta_0}{Y}))$\;
      $\agedges_a:=\varnothing$\;
      \ForEach{$l\xhookrightarrow{g:ch\,\alpha} l'$}{
        \ForEach{$\vec{c}\in \restr{d}{Y\times\agsym.\Loc}(B,l)$}{
          $g'=g[B=\vec{c}]$\;
          $\alpha':=\tau$\;
          \uIf{$g'=\bot$}{
            \Continue
          }
          \text{let }$\alpha=\alpha^{(1)}\ldots\alpha^{(m)}$\;
          \For{$i=1$ \KwTo $m$}{
             \ForEach{$f\in\mapf \wedge \agowner(f)=G$}{
               $X:=\Argsd(f)$\;
               $z:=\Argsr(f)$\;
               \uIf{ $x \in X\cap \text{LHS}(\alpha^{(i)})\neq \varnothing$}{
                $x':=\RHS(\alpha^{(i)})[X=\vec{c}]$\;
                $\alpha':=\alpha'.(z:=f(\vec{c}[x=x']))$\;
               }
             }
             \uIf{ $Y \cap \text{LHS}(\alpha^{(i)}) = \varnothing$}{
                $\alpha':=\alpha'.\alpha^{(i)}[Y=\vec{c}]$\;
             }
            }
          $\agedges_a := \agedges_a \cup \{ l\xhookrightarrow{g':ch\,\alpha'}l' \}$
        }
      }
      $G.\agedges := \agedges_a$\;
      \KwRet{$\agsym$}
    }
    \caption{Variable mapping abstraction}
    \label{alg:merge-abstraction}
  \end{algorithm} %

  The pseudocode for merging abstraction on an \agname $G$ is shown in \autoref{alg:merge-abstraction}.
}%

\subsection{Restricting the Scope of Abstraction}\label{sec:restricted-scope-abstraction}

The most general variant concerns a set $\mapf=\{(f_1,\Scope_1),\ldots,(f_m,\Scope_m)\}$, where each $f_i:\Eval(X_i)\mapsto\Eval(z_i)$ is a variable mapping in some graph $\agsym_j$, and $\Scope_i\subseteq\Loc_j$ defines its scope.
In the locations $l\in Sc_i$, mapping $f_i$ is applied and the values of all $v\in X_i$ are set to $v_0$ (depending on the needs, different value can be fixed per location).
Outside of $\Scope_i$, the variables in $X_i$ stay intact, and the new variable $z_i$ is assigned an arbitrary default value.

The resulting procedure is presented in Appendix~\ref{sec:general-abstraction}.

The abstract agent graph obtained by function \generalVarAbstraction from $\agsym$ in the context of $\magsym$ via $\mapf$ is denoted by $\abstr_{\mapf}(\agsym,\magsym)$.
The abstraction of the whole MAS graph $\magsym = \multiset{\Var_{sh},\agsym_1,\dots,\agsym_n}$ is defined (in accordance with \autoref{alg:abstraction-idea}) as
$$\abstr_{\mapf}(\magsym) = \multiset{\Var_{sh},\abstr_{\mapf}(\agsym_1,,\magsym)\dots,\abstr_{\mapf}(\agsym_n,\magsym)}$$

\begin{figure*}[t]\centering
\begin{subfigure}[c]{\columnwidth}
	\centering
	\includegraphics[width=\maxfitwidth]{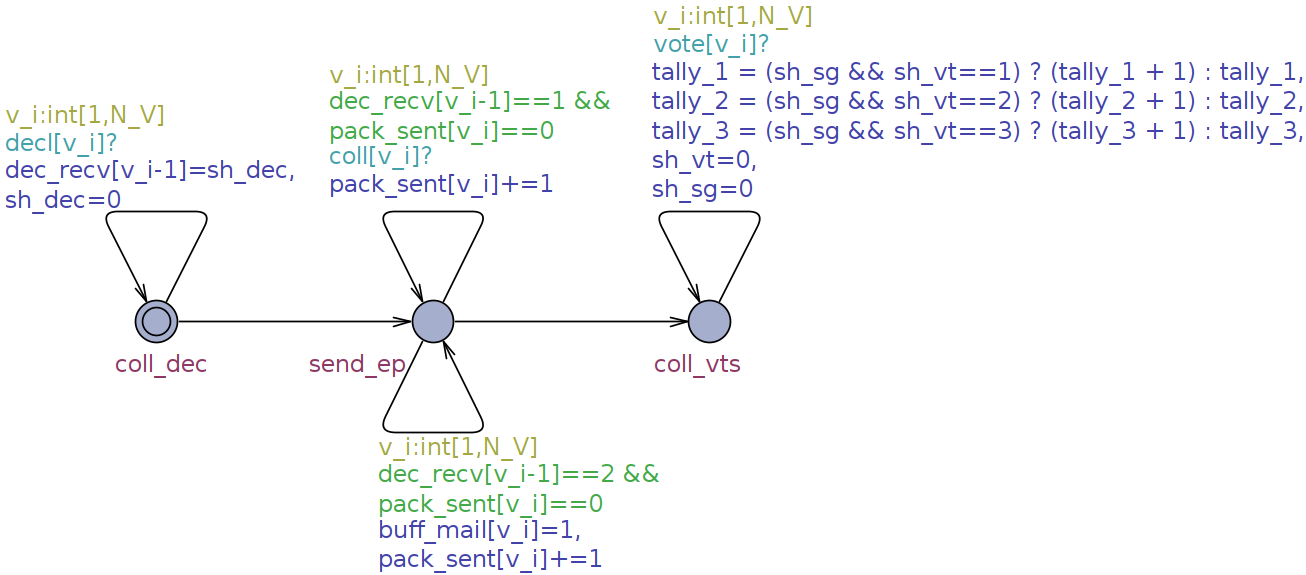}
	\caption{Election Authority graph}
	\label{fig:upp-simple-a}
\end{subfigure}
\hspace{0.6cm}
\begin{subfigure}[c]{\columnwidth}
	\centering
	\includegraphics[width=\maxfitwidth]{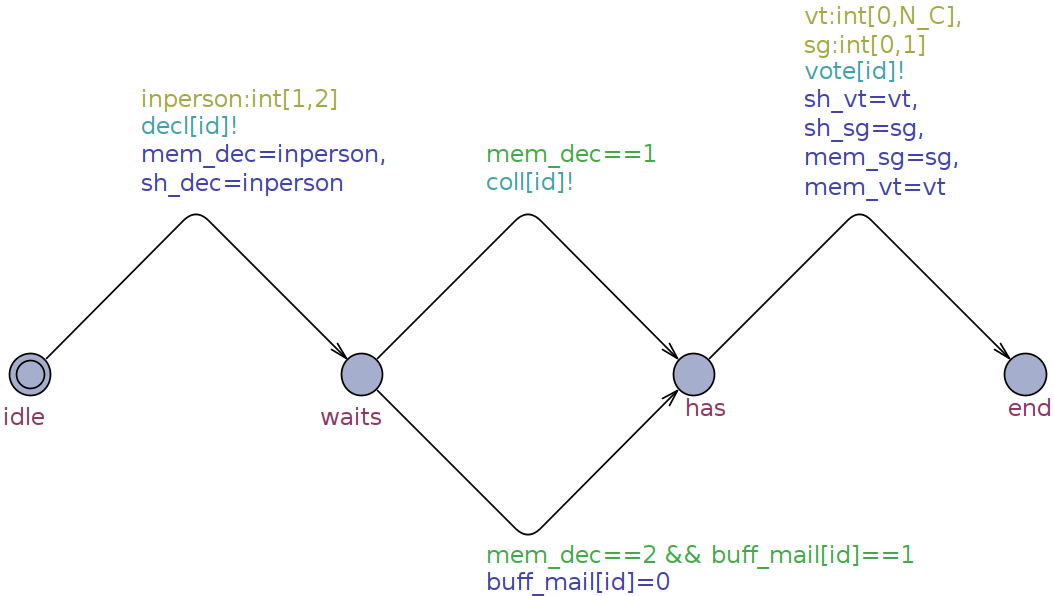}
	\caption{Voter graph}
	\label{fig:upp-simple-v}
\end{subfigure}
\caption{\magname\ for simplified postal voting}
\label{fig:upp-simple}
\end{figure*}

\newcommand{\scaletables}{0.84\textwidth}

\begin{table*}[t]\centering
\resizebox{\scaletables}{!}{%
\tiny

\begin{tabular}{l|rr|rrr|rrr|rrr}
	\toprule
	conf  & \multicolumn{2}{c|}{Concrete}               & \multicolumn{3}{c|}{Abstract 1}             & \multicolumn{3}{c|}{Abstract 2}             & \multicolumn{3}{c}{Abstract 3}                                                                                                                                    \\
	NV,NC & \#St                                        & tv (sec)                                 & ta (sec)                                 & \#St                                       & tv (sec)                                 & ta (sec) & \#St    & tv (sec) & ta (sec) & \#St    & tv (sec) \\
	\hline
	1,1   & 2.30e+1                                     & 0                                           & 0.03                                        & 1.90e+1                                    & 0                                           & 0.07        & 1.80e+1 & 0        & 0.16        & 1.60e+1 & 0           \\
	1,2   & 2.70e+1                                     & 0                                           & 0.03                                        & 2.10e+1                                    & 0                                           & 0.08        & 2.00e+1 & 0        & 0.06        & 1.70e+1 & 0           \\
	1,3   & 3.10e+1                                     & 0                                           & 0.03                                        & 2.30e+1                                    & 0                                           & 0.06        & 2.20e+1 & 0        & 0.05        & 1.80e+1 & 0           \\
	\hline
	2,1   & 2.41e+2                                     & 0                                           & 0.02                                        & 1.41e+2                                    & 0                                           & 0.06        & 1.26e+2 & 0        & 0.06        & 9.30e+1 & 0           \\
	2,2   & 3.69e+2                                     & 0                                           & 0.02                                        & 1.77e+2                                    & 0                                           & 0.04        & 1.66e+2 & 0        & 0.03        & 1.06e+2 & 0           \\
	2,3   & 5.29e+2                                     & 0                                           & 0.02                                        & 2.17e+2                                    & 0                                           & 0.06        & 2.14e+2 & 0        & 0.04        & 1.20e+2 & 0           \\
	\hline
	3,1   & 2.99e+3                                     & 0.01                                        & 0.02                                        & 1.14e+3                                    & 0                                           & 0.07        & 9.72e+2 & 0.01     & 0.05        & 5.67e+2 & 0           \\
	3,2   & 6.08e+3                                     & 0.01                                        & 0.02                                        & 1.62e+3                                    & 0.01                                        & 0.05        & 1.57e+3 & 0        & 0.04        & 6.93e+2 & 0           \\
	3,3   & 1.09e+4                                     & 0.04                                        & 0.02                                        & 2.20e+3                                    & 0.02                                        & 0.03        & 2.44e+3 & 0        & 0.05        & 8.38e+2 & 0.01        \\
	\hline
	4,1   & 3.98e+4                                     & 0.12                                        & 0.02                                        & 9.57e+3                                    & 0.05                                        & 0.08        & 7.94e+3 & 0.03     & 0.08        & 3.54e+3 & 0.02        \\
	4,2   & 1.06e+5                                     & 0.55                                        & 0.01                                        & 1.52e+4                                    & 0.08                                        & 0.08        & 1.60e+4 & 0.05     & 0.06        & 4.62e+3 & 0.04        \\
	4,3   & 2.36e+5                                     & 0.95                                        & 0.01                                        & 2.26e+4                                    & 0.12                                        & 0.08        & 2.99e+4 & 0.07     & 0.08        & 5.94e+3 & 0.06        \\
	\hline
	5,1   & 5.46e+5                                     & 1.48                                        & 0.02                                        & 8.17e+4                                    & 0.36                                        & 0.19        & 6.71e+4 & 0.18     & 0.25        & 2.23e+4 & 0.13        \\
	5,2   & 1.90e+6                                     & 6.42                                        & 0.02                                        & 1.43e+5                                    & 0.76                                        & 0.18        & 1.69e+5 & 0.5      & 0.23        & 3.09e+4 & 0.23        \\
	5,3   & 5.16e+6                                     & 24.95                                       & 0.02                                        & 2.30e+5                                    & 1.43                                        & 0.22        & 3.79e+5 & 1.16     & 0.22        & 4.21e+4 & 0.39        \\
	\hline
	6,1   & 7.58e+6                                     & 31.34                                       & 0.01                                        & 7.03e+5                                    & 4.39                                        & 0.55        & 5.79e+5 & 1.92     & 0.44        & 1.41e+5 & 0.92        \\
	6,2   & 3.41e+7                                     & 170.25                                      & 0.01                                        & 1.34e+6                                    & 10.87                                       & 0.50        & 1.82e+6 & 7.64     & 0.40        & 2.07e+5 & 1.83        \\
	6,3   & \multicolumn{2}{c|}{\tiny{\texttt{memout}}} & 0.01                                        & 2.31e+6                                     & 20.31                                      & 0.84                                        & 4.87e+6     & 22.67   & 0.40     & 2.97e+5     & 4.7                   \\
	\hline
	7,1   & \multicolumn{2}{c|}{\tiny{\texttt{memout}}} & 0.01                                        & 6.05e+6                                     & 46.75                                      & 2.34                                        & 5.07e+6     & 22.16   & 1.91     & 8.89e+5     & 8.34                  \\
	7,2   & \multicolumn{2}{c|}{\tiny{\texttt{memout}}} & 0.02                                        & 1.25e+7                                     & 149.84                                     & 1.33                                        & 1.98e+7     & 107.95  & 2.01     & 1.38e+6     & 16.11                 \\
	7,3   & \multicolumn{2}{c|}{\tiny{\texttt{memout}}} & 0.02                                        & 2.28e+7                                     & 304.86                                     & 2.49 &\multicolumn{2}{c|}{\tiny{\texttt{memout}}} & 2.35        & 2.08e+6 & 30.75                                          \\
	\hline
	8,1   & \multicolumn{2}{c|}{\tiny{\texttt{memout}}} & 0.02                                        & 5.20e+7                                     & 482.66                                     & 10.30 & \multicolumn{2}{c|}{\tiny{\texttt{memout}}} & 8.04        & 5.61e+6 & 66.44                                          \\
	8,2   & \multicolumn{2}{c|}{\tiny{\texttt{memout}}} & 0.19 & \multicolumn{2}{c|}{\tiny{\texttt{memout}}} & 12.17 &  \multicolumn{2}{c|}{\tiny{\texttt{memout}}} & 7.58                                      & 9.15e+6                                     & 150.86                                                                 \\
	8,3   & \multicolumn{2}{c|}{\tiny{\texttt{memout}}} & 0.07 & \multicolumn{2}{c|}{\tiny{\texttt{memout}}} & 9.52 & \multicolumn{2}{c|}{\tiny{\texttt{memout}}} & 7.99                                      & 1.44e+7                                     & 348.99                                                                 \\
	\hline
	9,1   & \multicolumn{2}{c|}{\tiny{\texttt{memout}}} & 0.12 &\multicolumn{2}{c|}{\tiny{\texttt{memout}}} & 70.49 & \multicolumn{2}{c|}{\tiny{\texttt{memout}}} & 64.96                                      & 3.53e+7                                     & 474.43                                                                 \\
	9,2   & \multicolumn{2}{c|}{\tiny{\texttt{memout}}} & 0.06 &\multicolumn{2}{c|}{\tiny{\texttt{memout}}} & 68.46 & \multicolumn{2}{c|}{\tiny{\texttt{memout}}} & 71.69  & \multicolumn{2}{c}{\tiny{\texttt{memout}}}                                                                                                                        \\
	9,3   & \multicolumn{2}{c|}{\tiny{\texttt{memout}}} & 0.12 &\multicolumn{2}{c|}{\tiny{\texttt{memout}}} & 70.09 & \multicolumn{2}{c|}{\tiny{\texttt{memout}}} & 72.72  & \multicolumn{2}{c}{\tiny{\texttt{memout}}}                                                                                                                        \\
	\bottomrule
\end{tabular}

}%
\caption{Experimental results for model checking of $\varphi_\mathit{bstuff}$ in over-approximations of postal voting}
\label{tab:benchmarks}
\end{table*}

\begin{table*}[t]\centering
\resizebox{\scaletables}{!}{%
\tiny

\begin{tabular}{l|rr|rrr|rrr|rrr}
\toprule
conf  & \multicolumn{2}{c|}{Concrete} & \multicolumn{3}{c|}{Abstract 1} & \multicolumn{3}{c|}{Abstract 2} & \multicolumn{3}{c}{Abstract 3} \\
NV,NC & \#St          & tv (sec)    & ta (sec)  & \#St     & tv (sec) & ta (sec)& \#St     & tv (sec) & ta (sec) & \#St      & tv (sec) \\
\hline
1,1   & 23            & 0          & 0.03     & 15          & 0       & 0.07      & 14       & 0       & 0.08       & 14        & 0       \\
1,2   & 27            & 0          & 0.03     & 15          & 0       & 0.05      & 14       & 0       & 0.06       & 14        & 0.01    \\
1,3   & 31            & 0          & 0.03     & 15          & 0       & 0.05      & 14       & 0       & 0.04       & 14        & 0       \\
\hline
2,1   & 241           & 0          & 0.01     & 81          & 0       & 0.04      & 70       & 0       & 0.04       & 70        & 0       \\
2,2   & 369           & 0          & 0.01     & 81          & 0       & 0.02      & 70       & 0       & 0.03       & 70        & 0       \\
2,3   & 529           & 0          & 0.03     & 81          & 0       & 0.02      & 70       & 0       & 0.04       & 70        & 0       \\
\hline
3,1   & 2987          & 0          & 0.01     & 459         & 0       & 0.03      & 368      & 0       & 0.03       & 368       & 0       \\
3,2   & 6075          & 0          & 0.02     & 459         & 0       & 0.03      & 368      & 0       & 0.03       & 368       & 0       \\
3,3   & 1.09e+4       & 0          & 0.02     & 459         & 0       & 0.03      & 368      & 0       & 0.03       & 368       & 0       \\
\hline
4,1   & 3.98e+4       & 0          & 0.01     & 2673        & 0       & 0.03      & 2002     & 0.01    & 0.04       & 2002      & 0       \\
4,2   & 1.06e+5       & 0          & 0.01     & 2673        & 0       & 0.05      & 2002     & 0       & 0.03       & 2002      & 0       \\
4,3   & 2.36e+5       & 0          & 0.01     & 2673        & 0       & 0.04      & 2002     & 0       & 0.03       & 2002      & 0       \\
\hline
5,1   & 5.46e+5       & 0          & 0.01     & 1.58e+4     & 0       & 0.04      & 1.11e+4  & 0       & 0.06       & 1.11e+4  & 0       \\
5,2   & 1.90e+6       & 0          & 0.01     & 1.58e+4     & 0       & 0.06      & 1.11e+4  & 0       & 0.05       & 1.11e+4  & 0       \\
5,3   & 5.16e+6       & 0          & 0.02     & 1.58e+4     & 0       & 0.07      & 1.11e+4  & 0       & 0.05       & 1.11e+4  & 0.01    \\
\hline
6,1   & 7.58e+6       & 0          & 0.01     & 9.40e+4     & 0       & 0.15      & 6.30e+4  & 0       & 0.09       & 6.30e+4  & 0       \\
6,2   & 3.41e+7       & 0.01       & 0.01     & 9.40e+4     & 0       & 0.14      & 6.30e+4  & 0       & 0.10       & 6.30e+4  & 0       \\
6,3   & 1.13e+8       & 0          & 0.01     & 9.40e+4     & 0       & 0.09      & 6.30e+4  & 0       & 0.09       & 6.30e+4  & 0       \\
\hline
7,1   & 1.06e+8       & 0          & 0.01     & 5.62e+5     & 0.01    & 0.28      & 3.60e+5  & 0       & 0.24       & 3.60e+5  & 0       \\
7,2   & $\gg$1e+8     & 0          & 0.01     & 5.62e+5     & 0       & 0.34      & 3.60e+5  & 0       & 0.21       & 3.60e+5  & 0       \\
7,3   & $\gg$1e+8     & 0.01       & 0.01     & 5.62e+5     & 0       & 0.35      & 3.60e+5  & 0       & 0.23       & 3.60e+5  & 0       \\
\hline
8,1   &  $\gg$1e+8    & 0          & 0.01     & 3.37e+6     & 0       & 0.90      & 2.08e+6  & 0       & 0.69       & 2.08e+6  & 0       \\
8,2   &  $\gg$1e+8    & 0          & 0.02     & 3.37e+6     & 0       & 1.03      & 2.08e+6  & 0       & 0.63       & 2.08e+6  & 0       \\
8,3   &  $\gg$1e+8    & 0          & 0.01     & 3.37e+6     & 0       & 0.86      & 2.08e+6  & 0       & 0.55       & 2.08e+6  & 0       \\
\hline
9,1   &  $\gg$1e+8    & 0          & 0.01     & 2.02e+7     & 0       & 4.41      & 1.21e+7  & 0       & 2.43       & 1.21e+7  & 0       \\
9,2   &  $\gg$1e+8    & 0          & 0.01     & 2.02e+7     & 0       & 2.80      & 1.21e+7  & 0       & 2.03       & 1.21e+7  & 0       \\
9,3   &  $\gg$1e+8    & 0          & 0.01     & 2.02e+7     & 0       & 2.69      & 1.21e+7  & 0       & 1.99       & 1.21e+7  & 0       \\
\hline
10,1  &  $\gg$1e+8    & 0          & 0.01     & 1.21e+8     & 0       & 9.61      & 7.03e+7  & 0       & 7.49       & 7.03e+7   & 0.01    \\
10,2  &  $\gg$1e+8    & 0          & 0.01     & 1.21e+8     & 0       & 7.83      & 7.03e+7  & 0       & 8.02       & 7.03e+7   & 0       \\
10,3  &  $\gg$1e+8    & 0          & 0.01     & 1.21e+8     & 0       & 8.99      & 7.03e+7  & 0       & 7.71       & 7.03e+7   & 0      \\
\bottomrule
\end{tabular}

}%
\caption{Experimental results for model checking of $\varphi_\mathit{dispatch}$ in under-approximations of postal voting}
\label{tab:benchmarks2}
\end{table*}

\section{\mbox{Correctness of Variable Abstraction}}\label{sec:correctness}

We will now prove that the abstraction scheme\extended{, proposed in Section~\ref{sec:abstraction},} preserves the formulae of \ACTLs if the computation of variable domain $d$ produces an upper-approximation of their reachable values.
In essence, we show that the abstraction always produces an over-approximation of the MAS graph, which induces an appropriate simulation relation, and thus guarantees (one-way) preservation of \ACTLs.

\subsection{Simulations between Models}\label{sec:simulations}

We first recall a notion of simulation between models\extended{, that preserves the satisfied formulae of \ACTLs}~\cite{Baier08mcheck,Clarke18principles,Cohen09abstraction-MAS}.

\begin{definition}\label{def:simulation-relation}
Let $M_i=(\textit{St}_i, I_i, \longrightarrow_i, \textit{AP}_i, L_i)$, $i=1,2$ be a pair of models, and let $\textit{AP}\subseteq \textit{AP}_1\cap\textit{AP}_2$ be a subset of common atomic propositions.
Model $M_2$ \emph{simulates} model $M_1$ over $\textit{AP}$ (written $M_1\simul_{\textit{AP}} M_2$) if there exists a \emph{simulation} relation $\mathcal{R}\subseteq \textit{St}_1\times\textit{St}_2$ over $\textit{AP}$, s.t.:
\begin{enumerate}[label=(\roman*)]
	\item for every $s_1\in I_1$, there exists $s_2\in I_2$ with $s_1\mathcal{R}s_2$;
	\item for each $(s_1,s_2)\in\mathcal{R}$:
	\begin{enumerate}[label=(\alph*)]
		\item $L_1(s_1)\cap\textit{AP} = L_2(s_2)\cap\textit{AP}$,\quad and
		\item if $s_1\to s'_1$ then there exists $s_2\to s'_2$ such that $s'_1\mathcal{R}s'_2$,
	\end{enumerate}
\end{enumerate}

Additionally, for a pair of reachable states $s_1,s_2$ in $M_1,M_2$ such that $(s_1,s_2)\in\mathcal{R}$, we say that the pointed model $(M_2,s_2)$ \emph{simulates} $(M_1,s_1)$ over $\textit{AP}$, and denote it by $(M_1,s_1)\simul_{\textit{AP}} (M_2,s_2)$.
\end{definition}
\extended{
  By $\mathcal{R}(s_1)=\{s_2\in St_2\ |\ (s_1,s_2)\in \mathcal{R}\}$ we denote the set of states in $M_2$ {simulating} state $s_1$, %
  and by $\mathcal{R}^{-1}(s_2)\{s_1\in St_1\ |\ (s_1,s_2)\in \mathcal{R}\}$ the set of states in $M_1$ {simulated by} $s_2$.%
}%

\begin{theorem}\label{prop:simulation-implies}
	For $(M_1,s_1)\simul_{\textit{AP}} (M_2,s_2)$ and any $\actl^\ast$ state formula $\psi$, built of propositions from $AP$ only, it holds that:
	\begin{equation}\label{eq:lemma-actl-ih}
		M_2,s_2\srel \psi\text{\quad implies\quad }M_1,s_1\srel \psi \tag{$\ast$}
	\end{equation}
\end{theorem}
\short{The proof is standard, see e.g.~\cite{Baier08mcheck\extended{,Clarke18principles}}.}
\extended{
  \begin{proof}
	Suppose that $(M_1,s_1)\simul_{\textit{AP}} (M_2,s_2)$, and $\psi$ is a state formula of $\actl^\ast$ using only propositions in $AP$.
	We will show that \eqref{eq:lemma-actl-ih} holds by induction over the structure of $\psi$.
	
	\textit{Induction basis:}
	Let $\psi = a$, where $a\in AP$.
	If $M_2,s_2\srel \psi$ holds, then $a\in L(s_2)$.
	Since $L_1(s_1)\cap\textit{AP} = L_2(s_2)\cap\textit{AP}$ by definition,
	it follows that $a\in L(s_1)$ and thus $M_1,s_1\srel a$.
	
	\textit{Induction step:}
	Suppose that
	\eqref{eq:lemma-actl-ih} holds for state formulae $\psi_1$, $\psi_2$. Then:
	\begin{enumerate}[noitemsep,topsep=0pt,parsep=0pt,partopsep=0pt,label=(\alph*)]
		\item If
		$M_2,s_2 \srel \psi_1 \wedge \psi_2$,
		then $M_2,s_2\srel \psi_1$ and $M_2,s_2\srel \psi_2$.
		Thus
		$M_1,s_1\srel \psi_1$ and $M_1,s_1\srel \psi_2$, therefore
		$M_1,s_1\srel \psi_1 \wedge \psi_2$.
		\item If
		$M_2,s_2\srel \psi_1 \vee \psi_2$,
		then $M_2,s_2\srel \psi_1$ or $M_2,s_2\srel \psi_2$.
		Thus
		$M_1,s_1\srel \psi_1$ or $M_1,s_1\srel \psi_2$, therefore
		$M_1,s_1\srel \psi_1 \vee \psi_2$.
		
		\item If $M_2,s_2\srel \Apath \varphi$, where $\varphi$ is some path formula, then
		$\forall \pi\in\textit{Paths}(s_2).\ M_2,\pi_2\srel \varphi$.
		Assume that $\exists \pi_1\in \textit{Paths}(s_1)\ .\ M_1,\pi_1\nsrel \varphi$
		and therefore $s_1\nsrel \Apath\varphi$.
		Notice that using (i) and (ii-b) of %
		\autoref{def:simulation-relation} we can inductively construct
		$\pi_2\in\textit{Paths}(s_2)$ corresponding to $\pi_1$, such that
		$\pi_1[i]\mathcal{R}\pi_2[i]$ for all
		$i=1,\ldots,\len(\pi_1)$.
		From (ii-b), all those pairs of states must satisfy the same
		set of atomic properties,
		it must be that $M_2,\pi_2\nsrel\varphi$ and thus $M_2,s_2\nsrel \Apath\varphi$, which
		is
		a contradiction.
	\end{enumerate}
  \end{proof}
}%

\subsection{Over-Approximations of MAS Graphs}\label{sec:overapprox}

Let $M_1=\unwrap(\magsym_1), M_2=\unwrap(\magsym_2)$ be models resulting from unwrapping of MAS graphs $\magsym_1,\magsym_2$.
We start with a notion of simulation between states and runs, and then use that to define what it means for a $M_2$ to be a may-approximation of $M_1$.

\begin{definition}\label{def:matching}%
State $s_2=\abracket{l_2, \eta_2}$ \emph{simulates} state $s_1=\abracket{l_1, \eta_1}$ with respect to variables $V\subseteq \Var_1\cap\Var_2$ (denoted $s_1\overapprox_V s_2$) iff $l_1=l_2$ and for all $v\in V$ we have ${\eta_1}(v)={\eta_2}(v)$.

Moreover, run $\pi_2\in\textit{Runs}(M_2)$ \emph{simulates} run $\pi_1\in \textit{Runs}(M_1)$ with respect to $V$ (denoted $\pi_1\overapprox_V \pi_2$) iff:
\begin{enumerate}[label=(\roman*)]
	\item $\len(\pi_1)=\len(\pi_2)=t$,\quad and
	\item for every $1\leq i\leq t$, it holds that $\pi_1[i]\overapprox_V \pi_2[i]$.
\end{enumerate}

\end{definition}

The following is straightforward.
\begin{lemma}\label{lemma:matching-rel-lift}
If $\pi_1\overapprox_V\pi_2$ and their length is $t$, then
for all $1\leq i<j\leq t$
it holds that $\pi_1[i,j]\overapprox_V\pi_2[i,j]$.
\end{lemma}

\begin{lemma}\label{lemma:over-approx-sim}
For pointed models $(M_1,s_1)$ and $(M_2,s_2)$, s.t. $s_1=\abracket{l,\eta_1}\in\textit{St}_1$ and $s_2=\abracket{l,\eta_2}\in\textit{St}_2$ it holds that:
\begin{center}
	$(M_1,s_1)\overapprox_V (M_2, s_2)$ implies $(M_1,s_1)\simul_{\mathit{AP(V)}} (M_2, s_2)$.
\end{center}
\end{lemma}

We emphasize that model's approximation is defined over the set of variables $V$, which, for its part, establishes the subset of common atomic propositions $AP(V)$ for the simulation relation.

\subsection{Variable Abstraction Is Sound}

We prove now that the abstraction method in Section~\ref{sec:abstraction}, based on upper-approximation of local domain, is indeed a simulation.

\begin{lemma}\label{lemma:over-approx-domain}
	Let $\magsym$ be a \magname and $d$ be an upper-approximation of a local domain for $B\subsetneq \Var$. Then, if state $\abracket{l,\eta}$ is reachable in $\mathcal{M}(\magsym)$, it must be that $\eta(B)\in d(B,l)$.
\end{lemma}

\begin{theorem}\label{lemma:over-approx}
\quad $\mathcal{M}(\magsym) \quad \overapprox_{\Var\setminus\Argsd(\mapf)}\quad \mathcal{M}(\abstr_{\mapf}^{\may}(\magsym))$.
\end{theorem}

\begin{proof}

Let $M_1=\mathcal{M}(\magsym)$ and $M_2=\mathcal{M}(\abstr_{\mapf}^{\may}(\magsym))$, $W=\Argsd(\mapf)$, $V=\Var_1\setminus W$, $\Var_2 = V\cup \Argsr(\mapf)$.
Notice that $\Argsr(\mapf)\cap V = \varnothing$ by definition of $\mapf$, thus applying the mapping functions from $\mapf$ on the evaluation $\eta$ will have no effect on its $V$ fragment, namely for any $X\subseteq\Var_1\cup\Var_2$ and $c\in\Dom(X)$ if $X\cap V = \varnothing$ then $\eta(V)=\eta[X=c](V)$.

The requirement (i) from \autoref{def:simulation-relation} is satisfied as the set of locations $\Loc$ remains unchanged and for any $\eta_1\in \textit{Sat}(g_0)$ and $\eta_2\in \textit{Sat}(g_0\wedge g_R)$, where $g_R\in \textit{Cond}(\Argsr)$, it holds that $\eta_1(V)=\eta_2(V)$.

By the construction of the $\abstr^{\may}$ each state $\abracket{l,\eta}\in \textit{St}_2$ is either initial $\abracket{l,\eta} \in I_2$ or $\exists \abracket{l',\eta'}\in \textit{St}_2$ such that $\eta = \Effect(\alpha, \eta'[W=c])$ for some $c\in d(W,l')$, $l'\xhookrightarrow{g:\alpha}l$, and $\eta'[W=c]\in \textit{Sat}(g)$.
By that, the requirement (ii) is naturally satisfied.
\end{proof}

Our main result follows directly from Theorems~\ref{prop:simulation-implies}, \ref{lemma:over-approx-sim}, and~\ref{lemma:over-approx}.
\begin{theorem}\label{prop:overapprox-soundness}
Let $\magsym$ be a MAS graph, and $\mapf$ a set of mappings as defined in Section~\ref{sec:restricted-scope-abstraction}. %
Then, for every formula $\psi$ of \ACTLs that includes no variables being removed or added by $\mapf$: \smallskip\\
\centerline{$\unwrap(\abstr_{\mapf}^{\may}(\magsym)) \models \psi$\quad implies \quad$\unwrap(\magsym) \models \psi.$}
\end{theorem}

\para{Under-approximating abstractions.}
An analogous result can be obtained for must-abstraction $\abstr_{\mapf}^{\must}(\magsym)$.
\begin{lemma}\label{lemma:under-approx-domain}
	Let $\magsym$ be a \magname and $d$ be a lower-approximation of a local domain for $B\subsetneq \Var$.
	Then, for any $c\in d(B,l)$ there must exist reachable $\abracket{l,\eta}$ such that $\eta(B)=c$.
\end{lemma}

\begin{theorem}\label{prop:underapprox-soundness}
\extended{Let $\magsym$ be a MAS graph, and $\mapf$ a set of mappings as defined in Section~\ref{sec:restricted-scope-abstraction}.
Then, f}\short{F}or each formula $\psi\in\ACTLs$ including no variables removed by $\mapf$:\quad
$\unwrap(\abstr_{\mapf}^{\must}(\magsym)) \not\models \psi$\ implies\ $\unwrap(\magsym) \not\models \psi.$
\end{theorem}

The proof is analogous (see appendix ~\ref{sec:must-abstr-proof}).

\para{Abstraction for MAS templates.}
Interestingly, when the MAS graph contains a subset of agent graphs being instances of the very same template (i.e., almost identical, up to variable renaming and evaluation of a constant parameters), one could apply an abstraction on that template fragment resulting with a coarser abstraction of the original MAS graph, yet exponentially faster to compute.

\begin{definition}
	MAS template $MT$, a compact representation of a \magname\ \magsym, as a tuple $(\Var_{sh},\textit{Const}_{sh},(GT_1, \#_1),\ldots, (GT_k, \#_k))$, which lists of the pairs of agent templates $GT_i$ and number of their instances $\#_i$, as well as the set of shared variables $\Var_{sh}$ and additionally the set of shared constant (read-only) variables $\textit{Const}_{sh}$.
\end{definition}

\begin{theorem}
Let $\magsym = \multiset{\Var_{sh},\agsym_1,\dots,\agsym_n}$ and $MT$ be a corresponding MAS template.
Then upper-approximating a local domain $d_i$ of an agent template $GT_i$ discarding all of the synchronisation labels from its edges, will induce a may-abstraction $\abstr^{\may}(MT)$.

Analogously, lower-approximating a local domain $d_i$ of an agent template $GT_i$ discarding all edges with synchronisation label, will induce a must-abstraction $\abstr^{\must}(MT)$.
\end{theorem}

\begin{proof}
	Follows directly from the fact that discarding synchronisation labels results with a coarser upper-approximation of a $d_i$ and discarding the edges with synchronisation labels with a coarser lower-approximation.
\end{proof}

\section{Case Study and Experimental Results}\label{sec:casestudy}

In this section, we evaluate our abstraction scheme on a simplified real-life scenario.

\subsection{Case Study: Integrity of Postal Voting}

As input, we use a scalable family of MAS graphs that specify a simplified postal voting system.
The system consists of a single agent graph for the Election Authority (depicted in \autoref{fig:upp-simple-a}) and $\NVot$ instances of eligible Voters (\autoref{fig:upp-simple-v}). 

Each voter can vote for one of the $\NCand$ candidates.
The voter starts at the location \texttt{idle}, and declares if she wants to receive the election package (that contains the voting declaration and a ballot) by post or to pick it up in person.
Then, the voter waits until it can be collected, which leads to location \texttt{has}.
At this point, the voter sends the forms back to the authority, either filled in or blank (e.g., by mistake).
The authority collects the voters' intentions (at location \texttt{coll\_dec}), distributes the packages (at \texttt{send\_ep}), collects the votes, and computes the tally (at \texttt{coll\_vts}).
A vote is added to the tally only if the declaration is signed and the ballot is filled.%

In the experiments for over-approximating abstraction, we have verified the formula\\
\centerline{$\varphi_\mathit{bstuff} \equiv \Apath\Always \prop{(\sum_{i=1}^{\NCand}tally[i] \leq \sum_{j=1}^{\NVot}pack\_sent[j] \leq \NVot)}$}\\
expressing the important property of resistance to \emph{ballot stuffing}.
More precisely, the formula says that the amount of sent packages can never be higher than the number of voters, and there will be no more tallied votes than packages.
The formula is satisfied in all the instances of our voting model.

In the experiments for under-approximation, we used\\
\centerline{$\varphi_\mathit{dispatch} \equiv \Apath\Always \prop{{\ coll\_vts}\textit{ imply }(\sum_{j=1}^{\NVot}pack\_sent[j] = \NVot)}$}\\
expressing that the election packages must be eventually dispatched to all the voters.\footnote{%
The aforementioned formula is equivalent to $\Apath\Sometm \prop{(\sum_{j=1}^{\NVot}pack\_sent[j] = \NVot)}$ in our model; the former variant is used due to \Uppaal's non-standard interpretation of the $\Apath\Sometm$.
}
The formula is false in all the considered instances of the system.

\subsection{Results of Experiments}

We have used the following abstractions:
\begin{itemize}
\item Abstraction 1: globally removes variables \texttt{mem\_sg} and \texttt{mem\_vt}, i.e., the voters' memory of the cast vote and whether the voting declaration has been signed;
\item Abstraction 2: removes the variable \texttt{mem\_dec} at locations \{\texttt{has}, \texttt{voted}\}, and \texttt{dec\_recv} at \{\texttt{coll\_vts}\};
\item Abstraction 3: combines the previous two.
\end{itemize}

Model checking of \extended{the \ACTLs formulae }$\varphi_\mathit{bstuff}$ and $\varphi_\mathit{dispatch}$ has been performed with 32-bit version of \Uppaal 4.1.24 on a laptop with Intel i7-8665U 2.11 GHz CPU, running Ubuntu 22.04.
The abstract models were generated using a script in node.js with default 2GB RAM memory limit.\footnote{
    Implementation prototype and utilized models can be found at \url{https://github.com/aamas23submission/postal-voting-model}.
}
\extended{The results were calculated by means of \texttt{verifyta} command line utility (v4.1.24), which is newer yet backwards compatible version which gives a more detailed summary.
Time measurements were taken using the external tool (node.js), %
therefore for small time estimates there could be deviations due to the `noises' caused by other factors.}
The results are presented in Tables~\ref{tab:benchmarks} and~\ref{tab:benchmarks2}.
Each row lists the scalability factors (i.e., the number of voters and candidates), the size and verification time for the original model (so called ``concrete model''), and the results for Abstractions 1--3.
``Memout'' indicates that the verification process ran out of memory. 
The columns `ta' and `tv' stand for abstract model generation and verification time accordingly.
In all the other cases, the verification of the abstract model was conclusive (i.e., the output was ``true'' for all the instances in Table~\ref{tab:benchmarks} and ``false'' for all the instances in Table~\ref{tab:benchmarks2}).

The results show significant gains. In particular, for the variant with $\NCand=3$ candidates, our over-approximating may-abstractions allowed to increase the main scalability factor by 3, i.e., to verify up to 9 instead of 6 voters.

In case when a $\Apath\Always$ formula is not satisfied by the model, its verification is in fact equivalent to finding a witness for a $\Epath\Sometm$ formula, which is often easy in practice.
And it becomes even more so if model checker utilizes on-the-fly techniques (as is the case of \Uppaal) and examines the model simultaneously with generation of states.
This is confirmed by experimental results for model checking of $\varphi_\mathit{dispatch}$.
However, it is worth noting that despite the lack of notable gains in terms of verification time from must-abstraction, the reduction in state space could be of immense importance when model is generated prior to its exploration.

\section{Conclusions}\label{sec:conclusions}

In this paper, we present a correct-by-design method for model reductions that facilitates model checking.
Theoretically speaking, our reductions are \emph{may/must} state abstractions.
Their main appeal, however, lies in their ease of use, rather than theoretical characteristics.
They come with a natural methodology, and require almost no technical knowledge from the user.
All that the user needs to do is to select a subset of variables to be removed from the MAS graph representing the system (in the simplest case).
The more involved approach relies on defining mappings that merge information stored in local variables of an agent module.

The methodology is agent-based, meaning that the abstractions transform the specification of the system at the level of agent graphs, without changing the overall structure.
Moreover, they are guaranteed to generate a correct abstract MAS graph, i.e., one that provides the lower (resp.~upper) bound for the truth values of formulae to be verified.
Note also that the procedure is generic enough to be used in combination with other techniques, such as partial-order reduction\extended{~\cite{Peled93representatives,Gerth99por,Jamroga20POR-JAIR}}.

We demonstrate the effectiveness of the method on a case study involving the verification of a postal voting procedure using the Uppaal model checker.
As we have shown, simple abstractions allow to expand the range of verification by many orders of magnitude.
Clearly, the efficiency of the method depends on the right selection of variables and the abstraction scope; ideally, this should be provided by a domain expert.

In the future, we want to combine variable abstraction with abstractions that transform locations in a MAS graph.
Even more importantly, we plan to extend the methodology from branching-time properties to formal verification of strategic ability.
An implementation as an extension of the STV model checker~\cite{Kurpiewski21stv-demo} is also considered.

\bibliographystyle{ACM-Reference-Format}
\bibliography{%
	bib/wojtek,%
	bib/wojtek-own,%
	bib/aamas23%
}

\clearpage
\appendix

\begin{algorithm}[!htp]
  \SetAlgoLined\DontPrintSemicolon
  \SetKw{KwGoTo}{go to}
  \SetKwFor{While}{while}{}{end while}%
  \nonl\myproc{\approxLocalDomain{$\agsym=\agsym_{\magsym}$, $V$}}{
    \ForEach{ $l \in \Loc$}{
      \label{alg:approx-ld-init1}
      $l.d := d_0$\; %
      $l.p:=\varnothing$\; %
      $l.\mathit{color}:=\textit{white}$\;
    }
    $l_0.d := \set{\eta(V)\mid \eta\in\textit{Sat}(g_0)}$\; %
    \label{alg:approx-ld-init2}
    $Q:=\varnothing$\;%
    \label{alg:approx-ld-q1}
    $\enqueue(Q,l_0)$\;
    \label{alg:approx-ld-q2}
    \label{alg:approx-ld-mainloop1}
    \While{$Q\neq \varnothing$}{
      $l:=\extractMax(Q)$\;%
      
      \tcp{process incoming edges}  
      \label{alg:approx-ld-procincedges1}
      $d':=l.d$\;
      \ForEach{$l'\in l.p$}{ %
          \ForEach{$l'\xhookrightarrow{g:\alpha} l$}{
              $l.d:=l.d\cup d'\auxop\procEdgeLabels(l',g,\alpha,l,V)$ \; 
          }
      }
      $l.p=\varnothing$\;
      \label{alg:approx-ld-pireset}
      \uIf{$d' \neq l.d$}{
          $l.\mathit{color}:=\textit{grey}$\;
          \label{alg:approx-ld-procincedges2}
      }
        \label{alg:approx-ld-procloops1}
        \tcp{process self-loops}  
        $d'':=l.d$\;\label{proc-self-loops}
        \ForEach{$l\xhookrightarrow{g:\alpha} l$}{
            $l.d:=l.d\cup \procEdgeLabels(l,g,\alpha,l,V)$ \;
        }
        
        \uIf{$d''\neq l.d$}{
          $l.\mathit{color}:=\textit{grey}$\;
          \KwGoTo \ref{proc-self-loops}\;
          \label{alg:approx-ld-procloops2}
        }
      \label{alg:approx-ld-enqueue1}
      \tcp{enqueue immeadiate-neighbours}
      \uIf{$l.\mathit{color}\neq\textit{black}$}{
        \ForEach{$l'\in\dSucc(l)$}{
          $Q:=\enqueue(Q,l')$\;
          $l'.p:=l'.p\cup\set{l}$\;
        }
        $l.\mathit{color}=\textit{black}$\;
        \label{alg:approx-ld-enqueue2}
      }
      \label{alg:approx-ld-mainloop2}
    }
    \KwRet{$\set{ \abracket{(V,l),l.d} \mid l\in\Loc }$}\;
    \label{alg:approx-ld-return}
  }

  \bigskip
  \setcounter{AlgoLine}{0}
  \nonl\myproc{\procEdgeLabels{$l$, $g$, $\alpha$, $l'$, $V$}}{
        $\delta_0 := \set{\eta\in\textit{Sat}(g) \mid \eta(V)\in l.d}$\;
        \text{let }$\alpha:=\alpha^{(1)}\ldots\alpha^{(m)}$\;
        \For{$i=1$ \KwTo $m$}{
          $\delta_i := \set{\eta' = \Effect(\alpha^{(i)},\eta) \mid \eta\in\delta_{i-1}}$
        }
        \KwRet{$\set{\eta(V) \mid \eta\in\delta_m}$}
  }

  \caption{\mbox{Approximation of local domain for $x\in\Var$}}
  \label{alg:variable-approx}
\end{algorithm}

 \begin{algorithm}[!htp]
  \SetAlgoLined\DontPrintSemicolon
  \caption{General abstraction}
  \label{alg:general-abstraction}
  \myproc{\generalVarAbstraction{$\agsym$, $d$, $\mapf$}}{
    $X:=\Argsd(\mapf)$\;
    $Z:=\Argsr(\mapf)$\;
    $\Scope:=\bigcup_{(f_i,\Scope_i)\in\mapf}\Scope_i$\;
    $g_0:=g_0\wedge(Z=f(\eta_0(X)))$\;
    $\agedges_a:=\varnothing$\;
    \ForEach{$l\xhookrightarrow{g:ch\,\alpha}l'$}{
      $F_1:=\{f_i\mid (f_i,\Scope_i)\in \mapf\wedge  l'\in \Scope_i\}$\; %
      $F_2:=\{f_i\mid (f_i,\Scope_i)\in \mapf\wedge  l'\in \Scope_i\}$\; %
      \uIf{$\set{l,l'}\cap\Scope=\varnothing$}{
        $\agedges_a := \agedges_a \cup \{ l\xhookrightarrow{g:ch\,\alpha}l' \}$\;
      }\uElse{
        \ForEach{$\vec{c}\in d(V,l)$}{
          $W_1:=\bigcup_{f\in F_1} \Argsd(f)$\;
          $W_2:=\bigcup_{f\in F_1} \Argsr(f)$\;
          
          $Y_1:=\bigcup_{f\in F_2} \Argsd(f)$\;
          $Y_2:=\bigcup_{f\in F_2} \Argsr(f)$\;
        
          $g':=g[W_1=\restr{\vec{c}}{W_1}]$\;

          $\alpha' := (W_2:= \eta_0(W_2).\alpha$\;
          $\alpha' := (W_1:=\restr{\vec{c}}{W_1}).\alpha'$\;

          $\alpha' := \alpha'.(Y_2:=F_2(\restr{\vec{c}}{Y_1}))$\;
          $\alpha' := \alpha'.(Y_1:= \eta_0(Y_1))$\;

          $\agedges_a := \agedges_a \cup \{ l\xhookrightarrow{g':ch\,\alpha'}l' \}$
        }
      }
    }
    $\agsym.\agedges:=\agedges_a$\;
    \KwRet{$\agsym$}
  }
\end{algorithm}%

\section{Approximating the Domains of Variables}\label{sec:approxDomain}

Given \magname\ \magsym\ over variables $\Var$, the approximation of reachable values for a set of variables $V\subseteq\Var$ takes the \extmagname $\agsym_{\magsym}$, and traverses it using a modified version of a priority-BFS algorithm~\cite{cormen2009introduction}.
It begins with the complement of the coarsest approximation of a local domain $d_0$, and starting from $l_0$ systematically explores locations of graph, and iteratively refines $d(V,l)$ with each visit of $l$ (until the stable approximation is obtained).
Each location $l$ must be visited at least once, and whenever some of its predecessors $l'$ gets approximation $d(V,l')$ refined, the location $l$ gets to be re-visited.

The max-priority queue $Q$ stores the locations that must be visited (possibly anew).
Within the queue, the higher traversal-priority is given for locations with greater reachability index $r(l)$, that is defined as the number of locations $l'\neq l$ reachable from $l$.\footnote{%
The values of $r(l)$ can be derived from reachability matrix produced by Warshall's algorithm, taking a row-wise sum of non-zero
elements apart from those in the main diagonal.} 
This will reduce the number of potential re-visits in comparison with generic FIFO variant.

In the algorithm we associate each location $l$ with attributes $\textit{color}\in\set{\textit{white}, \textit{grey}, \textit{black}}$, subset of relevant predecessors $p\subseteq \Loc\setminus\set{l}$ and its current local domain approximation $d$. 

In lines \ref{alg:approx-ld-init1}--\ref{alg:approx-ld-init2}, locations are initialized with \textit{white} color, an empty set of predecessors and initial approximation $d_0$.
The color indicates if the location is not visited yet (white), being visited and a finer approximation of its local domain was found (grey), or if it was visited (black). 
Lines \ref{alg:approx-ld-q1}--\ref{alg:approx-ld-q2} initialize the queue with location $l_0$. 
The while-loop of lines \ref{alg:approx-ld-mainloop1}--\ref{alg:approx-ld-mainloop2} describes the visit of a location $l$. 
The set of relevant predecessors indicates which had its approximation updated, that, for its part, may result in a `finer' $l.d$ (lines \ref{alg:approx-ld-procincedges1}--\ref{alg:approx-ld-procincedges2}).
After those were taken into account for $l.d$, the $l.p$ is reset (line \ref{alg:approx-ld-pireset}). 
Self-loops are processed separately with possible repetitions until stabilization (lines \ref{alg:approx-ld-procloops1}--\ref{alg:approx-ld-procloops2}).
The auxiliary procedure \procEdgeLabels simulates potential transitions by computing the image (restricted by $V$) of update $\alpha$ over evaluations satisfying guard $g$ and having their $V$ counterpart in source's $l.d$.
If during its visit, a location $l$ was colored grey (from either black or white), then all of its immediate neighbours $l'$ must be inspected independently of their color, adding $l$ as the relevant predecessor of each and lastly painting it black (lines \ref{alg:approx-ld-enqueue1}--\ref{alg:approx-ld-enqueue2}). 

The algorithm halts and returns the now stable approximation $d$ (line \ref{alg:approx-ld-return}) when the queue is empty and all locations are black.
Clearly, the termination is guaranteed by the finiteness of the variable domains.

\section{General variant of abstraction}\label{sec:general-abstraction}

For each pair $(f_i,\Scope_i)$ of mapping and its scope, it transforms the edges in $\agsym=\agsym_j$ in the following way:
\begin{itemize}
\item edges entering or inner the $\Scope_i$ have their actions appended with (1) update of the target variable $z_i$ and (2) update which sets the values of the source variables $X_i$ to their defaults (resetting those),
\item edges leaving or within the $\Scope_i$ an  have actions prepended with (1) update of source variables $X_i$ (a temporarily one to be assumed for original action) and (2) update which resets the values of the target variable $z_i$. 
\end{itemize}

\section{Proof for must-abstraction}\label{sec:must-abstr-proof}

Let $M_1=\mathcal{M}(\magsym)$ and $M_2=\mathcal{M}(\abstr_{\mapf}^{\must}(\magsym))$, $W=\Argsd(\mapf)$, $V=\Var_1\setminus W$, $\Var_2 = V\cup \Argsr(\mapf)$.

Now, we show that $M_1$ simulates $M_2$ over $AP(V)$ using the \autoref{def:simulation-relation}.

The requirement (i) is satisfied as the set of locations $\Loc$ remains unchanged and for any $\eta_1\in \textit{Sat}(g_0)$ and $\eta_2\in \textit{Sat}(g_0\wedge g_R)$, where $g_R\in \textit{Cond}(\Argsr)$, it holds that $\eta_1(V)=\eta_2(V)$.

By the construction of the $\abstr^{\must}$ each state $\abracket{l,\eta}\in \textit{St}_2$ is either initial $\abracket{l,\eta} \in I_2$ or $\exists \abracket{l',\eta'}\in \textit{St}_2$ such that $\eta = \Effect(\alpha, \eta'[W=c])$ for some $c\in d(W,l')$, $l'\xhookrightarrow{g:\alpha}l$, and $\eta'[W=c]\in \textit{Sat}(g)$.
Now we can utilize \autoref{lemma:under-approx-domain} and construct a sequence of corresponding states in $M_1$ up intill $\abracket{l,\eta[W=c]}$, which concludes with (ii) being satisfied.

\end{document}